\documentclass[runningheads]{llncs}

\usepackage{graphicx}
\usepackage{amssymb,wasysym}
\usepackage[linesnumbered,boxed]{algorithm2e}

\begin{document}

\title{On 1-factorizations of Bipartite Kneser Graphs}

\author{Kai Jin\inst{1}\orcidID{0000-0003-3720-5117}}

\authorrunning{K. Jin}

\institute{The Hong Kong University of Science and Technology, Hong Kong SAR
\email{cscjjk@gmail.com}}
\maketitle              


\newcommand{\PA}{\mathsf{PA}}
\newcommand{\CPA}{\mathsf{CPA}}

\begin{abstract}
It is a challenging open problem to construct an explicit 1-factorization of the bipartite Kneser graph $H(v,t)$,
  which contains as vertices all $t$-element and $(v-t)$-element subsets of $[v]:=\{1,\ldots,v\}$
  and an edge between any two vertices when one is a subset of the other.
In this paper, we propose a new framework for designing such 1-factorizations,
  by which we solve a nontrivial case where $t=2$ and $v$ is an odd prime power.
We also revisit two classic constructions for the case $v=2t+1$ ---
 the \emph{lexical factorization} and \emph{modular factorization}.
 We provide their simplified definitions and study their inner structures. As a result, an optimal algorithm is designed for computing the lexical factorizations. (An analogous algorithm for the modular factorization is trivial.)

\end{abstract}

\newcommand{\mod}{~\hbox{mod}~}

\section{Introduction}

The \emph{bipartite Kneser graph} $H(v,t)~(t<v/2)$ has as vertices all $t$-element and $(v-t)$-element subsets of $[v]:=\{1,\ldots,v\}$
  and an edge between any two vertices when one is a subset of the other.
Because it is regular and bipartite, each bipartite Kneser graph admits a 1-factorization due to Hall's Marriage Theorem \cite{Hall35}.
 (A 1-factor of a graph $G$ is a subgraph in which each node of $G$ has degree 1, and a 1-factorization of $G$ partitions the edges of $G$ into disjoint 1-factors.)
For the special case $v=2t+1$, the graph $H(2t+1,t)$ is also known as the \emph{middle level graph} and it admits two explicit 1-factorizations --
   the \emph{lexical factorization} \cite{LEXCIAL88} (see subsection~\ref{subsect:prel}) and \emph{modular factorization} \cite{MODULAR94} (see section~\ref{sect:modular}).
However, to the best of our knowledge, for decades it remains a challenging open problem to design explicit 1-factorizations for the general bipartite Kneser graphs.

\smallskip In this paper, we propose a natural framework to attack the open problem.
  Briefly, it attempts to find a special kind of 1-factorizations called resolvable 1-factorizations.
  We noticed that the lexical and modular factorizations and any 1-factorization of $H(2t+1,t)$ are resolvable.
  We also checked (by a \emph{C++} program) that there are no resolvable 1-factorization for $(v,t)=(6,2)$.
  Therefore, we can only expect for solving part of the open problem by using this framework.

    As our main result, Theorem~\ref{thm:main} states that finding a resolvable 1-factorization of $H(v,t)$ is equivalent to
      designing a special type of combinatorial designs, called perpendicular arrays \cite{TheoryPA,CRCHANDBOOk}.
      In particular, $\CPA(t,t+d,2t+d)$, where $d=v-2t$.
    According to this theorem and by using the known perpendicular arrays found in \cite{PA-rao,PA-3homo},
      we obtain the first resolvable 1-factorizations of $H(v,t)$ when $t=2$ and $t$ is an odd prime power or when $(t,d)\in \{3,8\},\{3,32\}$.
    On the other direction, we use the lexical and modular factorizations to
       obtain the first explicit constructions of $\CPA(t,t+1,2t+1)$, which are known to be existed in \cite{PA-larget}.

\medskip In addition to the construction of the new factorizations, we conduct a comprehensive study of the existing factorizations of the middle level graph mentioned above, which serves as part of an ongoing effort to solve the general case.

First, we unveil an inner structure of the lexical factorization, which leads to not only the first constructive proof for the fact that the lexical factorization is well-defined,
      but also an optimal algorithm for the following computational problem:
    Given $i$ and a $t$-element subset $A$, find the unique $A'$ such that $(A,A')$ belongs to the $i$-th 1-factor of the lexical factorization.
    The case $i=t+1$ of this problem was studied in \cite{Antipodal-code}. For $i\leq t$ it becomes more difficult and a trivial algorithm takes $O(v^2)$ time in the RAM model (where an atomic operation on a word accounts $O(1)$ time). We improve it to optimal $O(v)$ time (in section~\ref{sect:lexical}).
    (An $O(v)$ time algorithm for this problem on modular factorization is trivial.)

Second, we propose an intuitive definition of the modular factorization (in section~\ref{sect:modular}), which establishes an interesting connection between this factorization and the \emph{inversion number of permutations} (section 5.3 of \cite{book-KT}).
  As it is simpler than the original definition in most aspects, a few existing results about the modular factorization become more transparent with this new definition.

Also, we prove properties called \emph{variation laws} for the known 1-factorizations.

We will see the alternative definitions, inner structure, and variation laws are important for understanding the existing 1-factorizations. They have not been reported in literature and obtaining them requires nontrivial analysis.

\subsection{Motivation \& related work}\label{subsect:moti}
A 1-factor of the bipartite Kneser graph is also known as an antipodal matching in the subset lattice.
It is strongly related to the \emph{set inclusion matrix} introduced in \cite{inc-matrix1},
  which has connections to $t$-design in coding theory (see \cite{inc-matrix2,anotherchain} and the references within).
See \cite{Antipodal-code} for its another application in coding theory.

The 1-factorization problem of the middle level graph was motivated by the \emph{middle level conjecture}, which states that all the middle level graphs are Hamiltonian. It was hoped that people can find two 1-factors which form a Hamiltonian cycle \cite{LEXCIAL88}.
  Yet after extensive studies for thirty years the conjecture itself was settled by M\"{u}tze \cite{MIDLEVEL16};
  see also \cite{midlevel18} for a recent and shorter proof and see \cite{Mutze17soda} for an optimal algorithm for computing such a Hamiltonian cycle.
  Moreover, M\"{u}tze and Su \cite{MS17-H} settles the Hamiltonian problem for all the bipartite Kneser graphs.

\smallskip We give new applications of the 1-factorizations of $H(v,t)$ in hat-guessing games. We show that an optimal strategy in the unique-supply hat-guessing games can be designed from a 1-factorization of $H(v,t)$. To make the strategy easy to play, such a 1-factorization must be simple or at least admit an explicit construction. The details of this application are given in \ref{sect:hatguessing-motivation} due to space limits.

\subsection{Preliminaries}\label{subsect:prel}

\newcommand{\calP}{\mathcal{P}}
The \emph{subset lattice} is the family of all subsets of $[v]$, partially ordered by inclusion.
Let $\calP_t$ denote the $t$-th layer of this subset lattice, whose members are the $t$-element subsets of $[v]$.
Throughout the paper, denote $d=v-2t$. Let the words clockwise and counterclockwise be abbreviated as CW and CCW respectively.

\smallskip \noindent \textbf{\underline{A representation of the edges of $H(v,t)$}.\mbox{  }}
We identify each edge $(A,A')$ of $H(v,t)$ by a permutation $\rho$ of $t$ $\fullmoon$'s, $t$ $\triangle$'s, and $d$ $\times$'s: the (positions of) $t$ `$\fullmoon$'s indicate the $t$ elements in $A$; the $t$ `$\triangle$'s indicate the $t$ elements that are \textbf{not} in $A'$ (recall that $A'$ has $v-t$ elements); and the `$\times$'s indicate those in $A'-A$. We do not distinguish the edges with their corresponding permutations.

\smallskip Denote $[t \fullmoon, t \triangle, d \times]$ as the multiset of $2t+d$ characters with $t$ `$\fullmoon$'s, $t$ `$\triangle$'s, and $d$ `$\times$'s.
Giving a 1-factorization of $H(v,t)$ is equivalent to giving a \textbf{labeling function} $f$ from the ${2t+d}\choose {t,t,d}$ permutations of $[t \fullmoon, t \triangle, d \times]$ to $1,\ldots,{t+d \choose d}$ so that
\begin{enumerate}
\item[\textbf{(a)}] $f(\rho)\neq f(\sigma)$ for those pairs $\rho,\sigma$ who admit the same positions for $t$ $\fullmoon$'s; and
\item[\textbf{(b)}] $f(\rho)\neq f(\sigma)$ for those pairs $\rho,\sigma$ who admit the same positions for $t$ $\triangle$'s.
\end{enumerate}
If (a) and (b) hold, for fixed $i$, all edges labeled by $i$ constitute a 1-factor, denoted by $F_{f,i}$,
and $F_{f,1},\ldots,F_{f,{t+d\choose d}}$ constitute a 1-factorization of $H(v,t)$.

\newcommand{\flex}{f_{\mathsf{LEX}}}
\newcommand{\calL}{\mathcal{L}}

An example of the labelling function that satisfies (a) and (b) is given in \cite{LEXCIAL88}:

\medskip\noindent \textbf{\underline{The lexical factorization}\cite{LEXCIAL88}.\mbox{  }}
Let $\rho=(\rho_1,\ldots,\rho_{2t+1})$ be any permutation of $[t \fullmoon, t \triangle, 1 \times]$.
Arrange $\rho_1,\ldots,\rho_{2t+1}$ in a cycle in CW order.
For any $\rho_j$ that equals $\fullmoon$, it is \emph{positive} if there are strictly more $\fullmoon$'s than $\triangle$'s
in the interval that starts from the unique $\times$ and ends at $\rho_j$ in CW order.
The number of positive $\fullmoon$'s modular $t+1$ is defined to be
$\flex(\rho)$ (here, we restrict the remainder to $[t+1]$ by mapping 0 to $t+1$).
See Fig.~\ref{fig:flex}.
It is proved in \cite{LEXCIAL88} that $\flex$ satisfies the above conditions (a) and (b).
We provide in section~\ref{sect:lexical} a more direct proof for this.
The lexical factorization is $\{\calL_1,\ldots,\calL_{t+1}\}$, where $\calL_i=F_{\flex,i}$.

\begin{figure}[h]
  \centering
  \includegraphics[width=.6\textwidth]{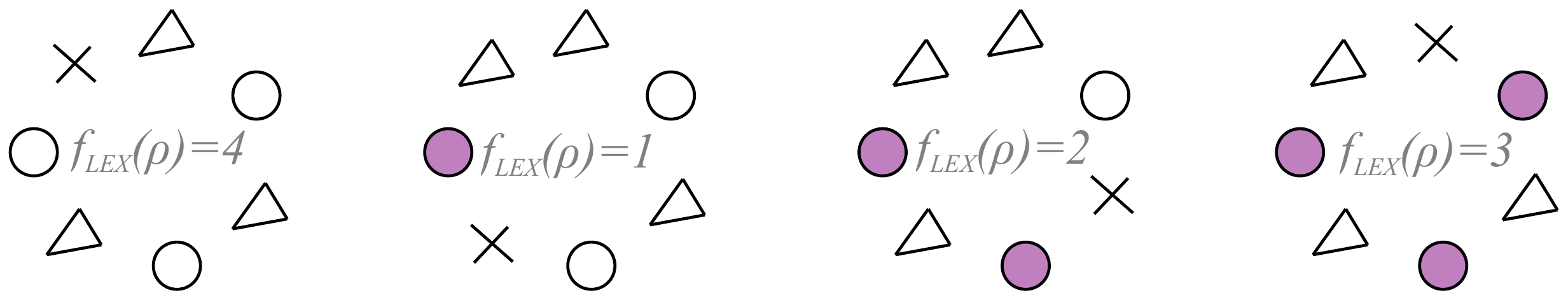}\\
  \caption{Illustration of the definition of $\flex$.
  In the graph, the solid circles indicate positive $\fullmoon$'s.
  Note that the positions of $\fullmoon$'s are identical in all the permutations drawn here.
  As we see, the four permutations are mapped to different numbers under $\flex$.}\label{fig:flex}
\end{figure}

\noindent \textbf{Note}: The original definition \cite{LEXCIAL88} of $\flex(\rho)$ actually calculates the number of nonnegative $\triangle$'s rather than positive $\fullmoon$'s.
For $\rho_j=\triangle$, it is said \emph{nonnegative} if there the number of $\fullmoon$'s is no less than the number of $\triangle$'s
 in the interval that starts from the unique $\times$ and ends at $\rho_j$ in CW order.
Nevertheless, it is clear that the number of nonnegative $\triangle$'s is the same as the number of positive $\fullmoon$'s.

\smallskip \noindent \textbf{Note}: The original definition \cite{LEXCIAL88} use $\calL_0$ to denote $\calL_{t+1}$.
In this paper, however, we choose $\calL_{t+1}$ instead of $\calL_0$ to make it consistent with the case $d>1$.
\newcommand{\calM}{\mathcal{M}}
\newcommand{\fm}{f_{\mathsf{mod}}}
\newcommand{\fmod}{f_{\mathsf{MOD}}}

\section{Construct ``resolvable'' 1-factorizations of $H(v,t)$}\label{sect:resolvable-construct}

This section introduces resolvable 1-factorizations of $H(v,t)$ and constructs some of them using combinatorial designs called perpendicular arrays (defined below).

\begin{definition}\label{def:f-gamma}
Assume $\{g_A \mid A \in \calP_t\}$ is a group of functions where $g_A$ is a bijection from $A^C=[v]-A$ to $[t+d]$ for every $A\in \calP_t$.
Assume $\gamma$ is a bijection from the set of all $d$-element subsets of $[t+d]$ to $1,\ldots, {t+d\choose d}$.
We define a labeling function $f_{\gamma,g}$ on the edges of $H(v,t)$ as follows: $f_{\gamma,g}(A,A'):=\gamma(g_A(A'-A))$.
\end{definition}

\noindent Note 1: Throughout, we use $g_A(X)$ to denote $\bigcup_{x\in X}g_A(x)$ for any $X\subseteq A^C$.

\noindent Note 2: Function $f_{\gamma,g}$ satisfies condition (a) trivially. Yet in most cases it does not satisfy condition (b) and hence does not define a 1-factorization of $H(v,t)$.

\begin{definition}\label{def:resolvable}
Let $f$ be the labelling function of a 1-factorization of $H(v,t)$.
We say $f$ is \emph{resolvable} if there are $\{g_A \mid A \in \calP_t\}$ and $\gamma$ as mentioned in Definition~\ref{def:f-gamma} such that
$f(A,A')\equiv \gamma(g_A(A'-A))$. In this case, we call $g_A$'s for $A \in \calP_t$ the \emph{resolved functions} of $f$
 and we say the 1-factorization defined by $f$ is \emph{resolvable}.
\end{definition}

\begin{remark}
Among other merits which make the resolvable 1-factorizations more interesting than the general ones,
  a resolvable 1-factorization takes only $(t+d)$ over ${t+d\choose d}$ fraction of storing space comparing to a general 1-factorization.
\end{remark}

The proofs of the following two lemmas are put into \ref{sect:resolvable-proof} due to space limits.

\begin{lemma}\label{lemma:some-resolvable-cases}
\begin{enumerate}
\item Any 1-factorization of $H(v=4,t=1)$ is resolvable.
\item Any 1-factorization of $H(2t+1,t)$, including the lexical factorization and modular factorization, is resolvable. (This claim is actually trivial.)
\item No 1-factorization of $H(6,2)$ is resolvable. (Will be proved by a program.)
\end{enumerate}
\end{lemma}

As shown by Lemma~\ref{lemma:some-resolvable-cases}, there could be $H(v,t)$'s without a resolvable 1-factorization,
hence we are not always able to design a resolvable 1-factorization of $H(v,t)$.
Nevertheless, the first two claims of Lemma~\ref{lemma:some-resolvable-cases} and the results given in the rest part of this section point out that for several cases we can do so.

\begin{lemma}\label{lemma:regardlessofcolor}
 Given $\{g_A \mid A \in \calP_t\}$ and $\gamma$ as above, the following are equivalent:
\begin{enumerate}
\item[(i)] function $f_{\gamma,g}(A,A')$ satisfies condition (b); and
\item[(ii)] When $A_1\neq A_2$ and $(A_1,A'_1),(A_2,A'_2)$ are two edges in $H(v,t)$, then $g_{A_1}(A'_1-A_1)= g_{A_2}(A'_2-A_2)$ implies that $A'_1\neq A'_2$.
\end{enumerate}
\end{lemma}

By Lemma~\ref{lemma:regardlessofcolor}, it is independent with the choice of $\gamma$ whether $f_{\gamma,g}(A,A')$ defines a 1-factorization of $H(v,t)$. Therefore, if we want to design a resolvable 1-factorization, the difficulty lies in and only lies in designing $\{g_A \mid A \in \calP_t\}$.\smallskip

By Lemma~\ref{lemma:regardlessofcolor} and Definition~\ref{def:f-gamma}, $H(v,t)$ has a resolvable factorization if and only if there exist resolved functions
$\{g_A \mid A\in \calP_t\}$
such that for $A_1\neq A_2$, $g_{A_1}(A'_1-A_1)= g_{A_2}(A'_2-A_2)$ implies $A'_1\neq A'_2$.
The following theorem shows that finding such functions is equivalent to designing some perpendicular arrays.

\smallskip A \emph{perpendicular array} \cite{TheoryPA,CRCHANDBOOk} with parameters $t,k,v$, denoted by $\PA(t,k,v)$, is a ${v \choose t}\times k$ matrix over $[v]$,
where each row has $k$ distinct numbers and each set of $t$ columns contain each $t$-element subset of $[v]$ as a row exactly once.

\medskip For $d\geq 0$, a $\PA(t,t+d,2t+d)$ is \emph{complete}, hence denoted by $\CPA(t,t+d,2t+d)$,
if each $(t+d)$-element subset of $[2t+d]$ is also contained in exactly one row.

\begin{theorem}\label{thm:main}
$H(2t+d,t)$ has a resolvable 1-factorization $\Leftrightarrow$ $\exists \CPA(t,t+d,2t+d)$.
\end{theorem}

\begin{proof}
$\Rightarrow$: Assume $f$ is the labeling function of a resolvable 1-factorization of $H(v=2t+d,t)$.
Then, $f(A,A')\equiv \gamma(g_A(A'-A))$ for some resolved functions $\{g_A \mid A\in \calP_t\}$ and
  a bijection $\gamma$ as mentioned in Definition~\ref{def:f-gamma}.

 We construct a matrix $M$ over $[v]$ as follows.
 For each $A\in \calP_t$, we build a row $(a^{(A)}_1,\ldots,a^{(A)}_{t+d})$ in $M$,
   where $a^{(A)}_i = g^{-1}_A(i)$ (which belongs to $A^C$ and thus belongs to $[v]$).
 As $\calP_t$ has $v\choose t$ elements, the size of matrix $M$ is $v\choose t$ by $k=t+d$.

We now verify that $M$ is a $\PA(t,t+d,2t+d)$. First, since $g^{-1}_A$ is bijective, $a^{(A)}_1,\ldots,a^{(A)}_{t+d}$ are distinct and so each row of $M$ contains $k=t+d$ distinct numbers. Next, for any $t$ columns $i_1,\ldots,i_t$,
we show that
\begin{equation}\label{eqn:row-distinct}
\{a^{(A_1)}_{i_1},\ldots,a^{(A_1)}_{i_t}\}\neq \{a^{(A_2)}_{i_1},\ldots,a^{(A_2)}_{i_t}\}
\end{equation}
for any distinct $A_1,A_2\in \calP_t$. Assume $\{j_1,\ldots,j_d\}=[t+d]-\{i_1,\ldots,i_t\}$.

Let $A'_1=A_1 \biguplus \{g^{-1}_{A_1}(j_1),\ldots,g^{-1}_{A_1}(j_d)\}$ and
 $A'_2=A_2 \biguplus \{g^{-1}_{A_2}(j_1),\ldots,g^{-1}_{A_2}(j_d)\}$.
Clearly, $g_{A_1}(A'_1-A_1)= \{j_1,\ldots,j_d\} = g_{A_2}(A'_2-A_2)$, thus $A'_1\neq A'_2$ by Lemma~\ref{lemma:regardlessofcolor}.
Thus $[v]-A'_1 \neq [v]-A'_2$.
Because $\{j_1,\ldots,j_d\}\biguplus \{i_1,\ldots,i_t\}=[t+d]$, we know
$A_1^C=\{g^{-1}_{A_1}(j_1),\ldots,g^{-1}_{A_1}(j_d)\}$ $\biguplus \{g^{-1}_{A_1}(i_1),\ldots,g^{-1}_{A_1}(i_t)\}$, which implies that $[v]-A'_1= \{g^{-1}_{A_1}(i_1),\ldots,g^{-1}_{A_1}(i_t)\}$.
Similarly, $[v]-A'_2= \{g^{-1}_{A_2}(i_1),\ldots,g^{-1}_{A_2}(i_t)\}$.
Altogether, $\{g^{-1}_{A_1}(i_1),\ldots,g^{-1}_{A_1}(i_t)\}\neq \{g^{-1}_{A_2}(i_1),\ldots,g^{-1}_{A_2}(i_t)\}$, i.e., (\ref{eqn:row-distinct}) holds.

Next, we argue that $M$ is a $\CPA(t,t+d,2t+d)$. This reduces to proving that each row of $M$ forms
  a distinct $(t+d)$-element subset of $[2t+d]$, which follows from the fact that the row constructed from $A$ is a permutation of $A^C$.

\medskip \noindent $\Leftarrow$: Assume $M$ is a $\CPA(t,t+d,2t+d)$.
First, we construct $\{g_A \mid A\in \calP_k\}$.
For any row $(a_1,\ldots,a_{t+d})$ of $M$, assuming that $A^C=\{a_1,\ldots,a_{t+d}\}$, define $g_A(a_i)=i$ for $i \in [t+d]$.
Obviously, each $g_A$ for $A\in \calP_k$ is defined exactly once.

Below we verify that when $A_1\neq A_2$, equality $g_{A_1}(A'_1-A_1)=g_{A_2}(A'_2-A_2)$ would imply $A'_1\neq A'_2$.
According to Lemma~\ref{lemma:regardlessofcolor}, this further implies that for any $\gamma$ as mentioned in Definition~\ref{def:f-gamma}, $f_{\gamma,g}(A,A')$ is a labeling function satisfying conditions (a) and (b), and hence $H(2t+d,t)$ has a resolvable 1-factorization.

\smallskip Suppose to the opposite that $g_{A_1}(A'-A_1)=g_{A_2}(A'-A_2)=\{j_1,\ldots,j_d\}$.
Assume $\{i_1,\ldots,i_t\}=[t+d]-\{j_1,\ldots,j_d\}$.
Because $g_{A_1}(A'-A_1)=\{j_1,\ldots,j_d\}$, we know $g_{A_1}([v]-A')=\{i_1,\ldots,i_t\}$,
so $[v]-A'=\{g^{-1}_{A_1}(i_1),\ldots,g^{-1}_{A_1}(i_t)\}$.
Similarly, because $g_{A_2}(A'-A_2)=\{j_1,\ldots,j_d\}$, we get $[v]-A'=\{g^{-1}_{A_2}(i_1),\ldots,g^{-1}_{A_2}(i_t)\}$.
Moreover, because $M$ is a $\PA(t,t+d,2t+d)$ where
$\{g^{-1}_{A_1}(i_1),\ldots,g^{-1}_{A_1}(i_t)\}$ and $\{g^{-1}_{A_2}(i_1),\ldots,g^{-1}_{A_2}(i_t)\}$ appear in two rows of $M$ in the columns indexed by $i_1,\ldots,i_t$, these two sets are distinct. Thus $[v]-A'\neq [v]-A'$. Contradiction.  \qed
\end{proof}

\subsection{Applications of Theorem~\ref{thm:main}}

\begin{lemma}\label{lemma:PA-exists}
\begin{enumerate}
\item For $t=1$, there is always a $\PA(t,2t+d,2t+d)$. (trivial)
\item \cite{PA-rao} For $t=2$ and an odd prime power $2t+d$, there is a $\PA(t,2t+d,2t+d)$.
\item \cite{PA-3homo} For $t=3$ and $2t+d\in \{8,32\}$, there is a $\PA(t,2t+d,2t+d)$.
\end{enumerate}
\end{lemma}

The following lemma is trivial and its proof can be found in \ref{sect:resolvable-proof}.
\begin{lemma}\label{lemma:PA-CPA}
Any $t+d$ columns of a $\PA(t,2t+d,2t+d)$ form a $\CPA(t,t+d,2t+d)$.
\end{lemma}

Lemma~\ref{lemma:PA-CPA} points out a way to construct a $\CPA(t,t+d,2t+d)$. Yet it is unknown whether every $\CPA(t,t+d,2t+d)$ can be constructed this way.
We conjecture so. If so, finding resolvable 1-factorizations reduces to finding $\PA(t,2t+d,2t+d)$'s.

The following is a corollary of Lemma~\ref{lemma:PA-exists}, Lemma~\ref{lemma:PA-CPA}, and Theorem~\ref{thm:main}.
\begin{corollary}
Graph $H(2t+d,t)$ has a resolvable 1-factorization when
1. $t=1$, or 2. ($t=2$ and $2t+d$ is an odd prime power), or 3. ($t=3$ and $2t+d\in \{8,32\}$).
\end{corollary}

The constructions of $\PA(t,2t+d,2t+d)$ for those pairs of $(t,d)$ discussed in Lemma~\ref{lemma:PA-exists} are explicit and quite simple (see \cite{PA-rao,PA-3homo}).
Also, our construction of the resolvable 1-factorization of $H(2t+d,t)$ using a $\CPA(t,t+d,2t+d)$ is extremely simple (as shown in the proof of Theorem~\ref{thm:main}). As a result, the resolvable 1-factorizations of $H(2t+d,t)$ mentioned in this corollary are explicit and simple.

\smallskip Perpendicular arrays have not been studied extensively in literature.
In addition to the existence results mentioned in Lemma~\ref{lemma:PA-exists}, there do exist $\PA(3,5,5)$ and $\PA(t,t+1,2t+1)~(t\geq 1)$ and some other perpendicular arrays.
Yet the construction of $\PA(t,t+1,2t+1)$ (in \cite{PA-larget}) is not explicit and thus not too useful (regarding that we are only interested in explicit factorizations of $H(2t+1,t)$). A $\PA(3,5,5)$ is also useless to us since $5<2\times 3$.
Because a $\CPA(t,t+d,2t+d)$ automatically implies a resolvable 1-factorization of $H(v,t)$,
  we hope that our results motivate more study on the perpendicular arrays in the future.

\paragraph{Another application of Theorem~\ref{thm:main} --- construction of $\CPA(t,t+1,2t+1)$.}
As shown in Lemma~\ref{lemma:some-resolvable-cases}.2, the lexical and modular factorization of $H(2t+1,t)$ are both resolvable.
The resolved functions of $\flex$ and $\fmod$ will be demonstrated in the next sections.
Using these resolved functions and applying the proof of Theorem~\ref{thm:main}, we can easily construct two $\CPA(t,t+1,2t+1)$s.
Therefore, as byproducts, we obtain (the first) explicit constructions of (complete) $\PA(t,t+1,2t+1)$
  (note that \cite{PA-larget} only showed the existence of $\PA(t,t+1,2t+1)$).

\section{Revisit the lexical factorization}\label{sect:lexical}

Recall $\flex$ in subsection~\ref{subsect:prel}, which is a labeling function of $H(2t+1,t)$.
In this section, we first give $\{g_A\}$ of $\gamma$ so that $\flex=f_{\gamma,g}$.
Based on this formula we then show that $\flex$ satisfies (a) and (b) and thus that it indeed defines a 1-factorization.
Moreover, by applying $\flex=f_{\gamma,g}$, we design optimal algorithms for solving
  two fundamental computational problems about this factorization (P1 and P2 below).
Finally, we introduce a group of \emph{variation laws} of $\flex$.

\begin{enumerate}
\item[P1.] Given $A\in \calP_t$ and $i\in \{1,\ldots,t+1\}$, how do we find the unique $A'$ so that $(A,A')\in \calL_i$?
    In other words, given number $i$ and the positions of $\fullmoon$'s in $\rho$ and suppose $\flex(\rho)=i$,
        how do we determine the position of $\times$ in $\rho$?\medskip
\item[P2.] Given a $A'\in \calP_{t+1}$ and $i\in \{1,\ldots,t+1\}$, how do we find the unique $A$ so that $(A,A')\in \calL_i$?
    In other words, given number $i$ and the positions of $\triangle$'s in $\rho$ and suppose $\flex(\rho)=i$ ,
        how do we determine the position of $\times$ in $\rho$?
\end{enumerate}

\subsection{Preliminary lemmas}\label{subsect:flex-pre}

The two lemmas given in this subsection are trivial; proofs can be found in \ref{sect:lexical-proof}.

\smallskip Given $S=(s_1,\ldots,s_v)$, the $j$-th $(0\leq j<v)$ \emph{cyclic-shift} of $S$ is
$S^{(j)}:=(s_{1+j},\ldots,s_{v+j})$, where subscripts are taken modulo $v$ (and restricted to $[v]$).

\begin{lemma}\label{lemma:cyclic-unique}
Given any sequence $S$ of $t$ of right parentheses `$)$' and $t+1$ left parentheses `$($'.
There exists a unique cyclic-shift $S^{(j)}$ of $S$ whose first $2t$ parentheses are paired up when parenthesized,
and we can compute $j$ in $O(t)$ time.
\end{lemma}

\begin{example}\label{example:S}
Assume $t=9$, $S=(_1(_2)_3)_4)_5(_6(_7(_8)_9)_{10}(_{11})_{12})_{13}(_{14}(_{15})_{16}(_{17}(_{18})_{19}$.

The unique cyclic-shift in which the first $2t$ parentheses are paired up is:
$$S^{(14)}=\bigg(_{15}\bigg)_{16}\bigg(_{17}\Big(_{18}\Big)_{19}\Big(_1(_2)_3\Big)_4\bigg)_5\bigg(_6\Big(_7(_8)_9\Big)_{10}\Big(_{11}\Big)_{12}\bigg)_{13}\bigg(_{14}.$$
\end{example}

\begin{definition}
Given $S=(s_1,\ldots,s_{2t+1})$, $t$ of which are ')' and $t+1$ are '('.
It is said \emph{canonical} if its first $2t$ parentheses are paired up when parenthesized.
\end{definition}

\newcommand{\ind}{\mathsf{index}}
\newcommand{\dep}{\mathsf{depth}}
\begin{definition}[Indices of the $2t+1$ parentheses]\label{def:index} For any canonical parentheses sequence $S$,
we index the $t+1$ left parentheses in $S$ by $0,\ldots,t$ according to the following rule:
\textbf{The smaller the depth, the less the index; and index from right to left for those under the same depth.}
Here, depth is defined in the standard way; it is the number of pairs of matched parentheses that cover the fixed parenthesis.
Moreover, \textbf{we index the $t$ right parentheses in such a way that any two paired parentheses have the same index.}
\end{definition}

For $S^{(14)}$ above, the depth and index are shown below (index on the right).

\begin{scriptsize}
$$\underbrace{\bigg(                              \bigg)}_0
    \underbrace{\bigg( \underbrace{\Big(\Big)}_1   \underbrace{\Big(\underbrace{()}_2\Big)}_1       \bigg)}_0
    \underbrace{\bigg( \underbrace{\Big(\underbrace{()}_2\Big)}_1 \underbrace{\Big(\Big)}_1      \bigg)}_0
    \underbrace{\bigg( }_0;\quad
~~ \underbrace{\bigg(                              \bigg)}_3
    \underbrace{\bigg( \underbrace{\Big(\Big)}_7   \underbrace{\Big(\underbrace{()}_9\Big)}_6       \bigg)}_2
    \underbrace{\bigg( \underbrace{\Big(\underbrace{()}_8\Big)}_5 \underbrace{\Big(\Big)}_4      \bigg)}_1
    \underbrace{\bigg( }_0.$$
\end{scriptsize}

This definition of indices is crucial to the next lemma and the entire section.

For convenience, denote by $\dep(s_i),\ind(s_i)$ the depth and index of $s_i$.

\begin{lemma}\label{lemma:index-stronger}
When $S$ is canonical, for any $s_l=($ and $s_r=)$,
  there are more $)$'s than $($'s in the cyclic interval $\{s_{l+1},\ldots,s_{r}\}$ if and only if $\ind(s_l)\geq \ind(s_r)$.
\end{lemma}

\subsection{Finding resolved functions $\{g_A\}$ of $\flex$}\label{subsect:g_A-flex}

\textbf{\underline{Parenthesis representation}.}
We can represent any $A\subseteq[v]$ by a sequence of parentheses $S=(s_1,\ldots,s_v)$ where
$s_x=')'$ if $x\in A$ and $s_x='('$ if $x\notin A$.
For example, $A=\{3,4,5,9,10,12,13,16,19\}$ is represented by the $S$ given in Example~\ref{example:S} above.
Notice that if $A\in \calP_t$, its associate sequence $S$ contains $t$ ')'s.

\begin{definition}\label{def:g_A-flex}
Fix $A\in \calP_t$ and let $S$ denote its parentheses sequence.
We abuse $\ind(s_x)$ to mean the index of $s_x$ in the unique canonical cyclic-shift of $S$ (uniqueness is by Lemma~\ref{lemma:cyclic-unique}).
For any $x\in A^C$ (hence $s_x='('$), define $g_A(x):= \ind (s_x) \mod (t+1) (\in [t+1])$ (restrict to $[t+1]$ by mapping $0$ to $t+1$).
\end{definition}

Because left parentheses have distinct indices, $g_A$ is a bijection as required.

\begin{theorem}
Let $\gamma$ be the natural bijection from all the $1$-element subsets of $[t+1]$ to $[t+1]$, which maps $\{x\}$ to $x$.
Define $\{g_A \mid A\in \calP_t\}$ as in Definition~\ref{def:g_A-flex}. Then,  $\flex=f_{\gamma,g}$. In other words, $\flex(A,A\cup \{x\})\equiv g_A(x)~(x\in A^C)$.
\end{theorem}

\begin{proof}
Build the parentheses sequence $S=(s_1,\ldots,s_{2t+1})$ of $A$ and the permutation $\rho=(\rho_1,\ldots,\rho_{2t+1})$ of $[t \fullmoon, t \triangle, 1 \times]$  corresponding to edge $(A,A\cup \{x\})$.
Recall that $\flex(A,A\cup \{x\}):=p \mod (t+1) \in [t+1]$, where $p$ is the size of
$P=\{\rho_r = \fullmoon \mid \hbox{there are more $\fullmoon$s than $\triangle$s in the cyclic interval } (\rho_{x+1},\ldots,\rho_r)\}$.
Observe that $S$ can be constructed from $\rho$ by replacing $\fullmoon,\triangle,\times$ to ')','(','('.
So, $\{s_r = ')' \mid \hbox{there are more $')'$s than $'('$s in the cyclic interval } (s_{x+1},\ldots,s_r)\}$,
  which equals $\{s_r = ')' \mid \ind(s_x)\geq \ind(s_r) \}$ by Lemma~\ref{lemma:index-stronger} (indices refer to those in the canonical cyclic-shift of $S$),
  has the same size as $P$, so $\ind(s_x)=p$.
Further by Definition~\ref{def:g_A-flex}, $g_A(x)=\ind(s_x) \mod (t+1) (\in [t+1])=\flex(A,A\cup \{x\})$.
\end{proof}

\begin{theorem}\label{thm:flex-a-b}
$\flex$ satisfies conditions (a) and (b).
\end{theorem}
\begin{figure}[b]
  \centering \includegraphics[width=.38\textwidth]{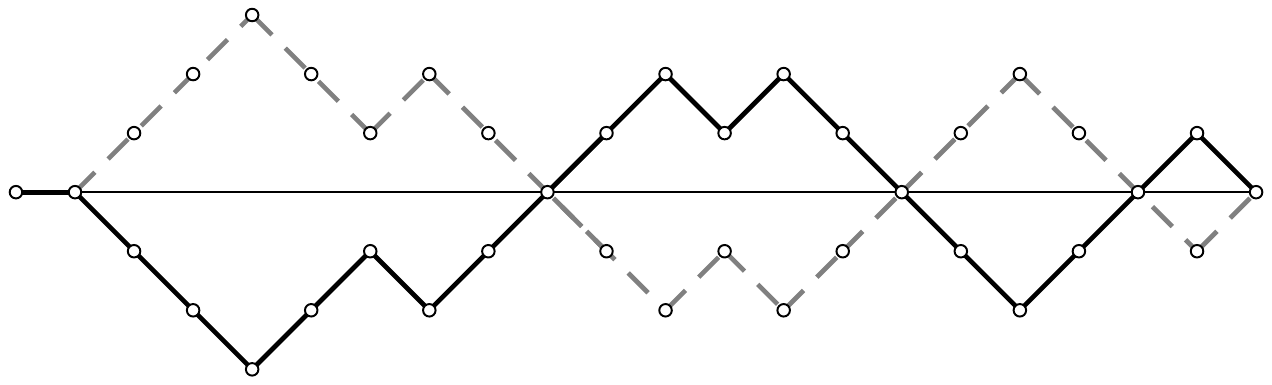}\\
  \caption{$\flex(\rho^*) \mod (t+1)+\flex(\rho)\mod (t+1)=t$. The dashed line indicates $\rho^*$.}\label{fig:positivedual}
\end{figure}
\begin{proof}
Because $\flex$ equals $f_{\gamma,g}$, applying Note~2 below Definition~\ref{def:f-gamma}, this labeling function satisfies condition(a).
Below we prove that it also satisfies condition (b).

\smallskip Define the \emph{dual} of $\rho$, denoted by $\rho^*$, to be another permutation of $[t \fullmoon, t \triangle, 1 \times]$ which is constructed from $\rho$ by swapping the $\triangle$'s with $\fullmoon$'s.
As illustrated in Fig.~\ref{fig:positivedual}, we have (i): $\flex(\rho^*)\mod (t+1)+\flex(\rho)\mod (t+1)=t$ for any $\rho$ .

Consider $t+1$ distinct permutations $\rho^0,\ldots,\rho^t$ sharing the same positions of $\triangle$'s.
Then, $(\rho^0)^*,\ldots,(\rho^t)^*$ share the same positions of $\fullmoon$'s.
Using condition (a), $\flex((\rho^0)^*),\ldots,\flex((\rho^t)^*)$ are distinct.
So $t-\flex((\rho^0)^*)\mod (t+1),\ldots,t-\flex((\rho^t)^*)\mod(t+1)$ are distinct.
 So $\flex(\rho^0)\mod (t+1),\ldots,\flex(\rho^t)\mod(t+1)$ are distinct by (i),
 i.e., $\flex(\rho^0),\ldots,\flex(\rho^t)$ are distinct. Thus (b) holds.\qed
\end{proof}

\begin{remark}
In the original proof of Theorem~\ref{thm:flex-a-b} in \cite{LEXCIAL88}, it proves the existence of bijections $g_A$'s $(A\in \calP_t)$ such that $\flex=f_{\gamma,g}$, yet how to define such $g_A$'s is neither explicitly given, nor implicitly given.
As we have seen in Definition~\ref{def:index}, giving this definition is not easy, even though the definition of $\flex$ is known.

There are two advantages of having explicit $\{g_A\}$.
First, the ideas we used in defining $g_A$ could be useful in finding resolvable 1-factorizations for the case $v>2t+1$.
Second, to solve P1 and P2 (in the next subsection),
  it seems necessary to have an explicit definition of $\{g_A\}$ for the efficiency of computation.
\end{remark}

\subsection{Linear Time Algorithms for P1 and P2}\label{subsect:alg}

Problem P1 admits a trivial $O(t^2)$ time solution as follows.
Given the positions of $\fullmoon$'s in $\rho$ and the number $i$, we can enumerate the position of the unique $\times$ among the remaining $t+1$ positions and compute $\flex(\rho)$ in $O(t)$ time, until that
  the computed value is $i$. Problem P2 can be solved symmetrically.

\smallskip Applying the results in subsection~\ref{subsect:g_A-flex}, we can solve P1 much more efficiently.
  Briefly, using those indices of parentheses in Definition~\ref{def:index}, we can compute $\flex()$ for all permutations $\rho^0,\ldots,\rho^t$ in which the positions of $\fullmoon$'s are as given altogether, and then find $\rho^j$ so that $\flex(\rho^j)=i$. See the details in Algorithm~\ref{alg:flexP1}.

\begin{algorithm}[h]
\caption{Computing the unique $A'$ such that $(A,A')\in \calL_i$.}\label{alg:flexP1}
\KwIn{A set $A\in \calP_t$ and a number $i\in [t+1]$.}
\KwOut{The set $A'=A\cup \{z\}$ so that $(A,A')\in \calL_i$.\\ \qquad (Integer $z$ indicates the position of $\times$ so that $\flex(\rho)=i$.)}
Compute the parentheses sequence $S$ of $A$.\\
Compute the unique $j$ so that the first $2t$ parentheses are paired up in $S^{(j)}$.\\
Compute the indices of all parentheses in $S'=S^{(j)}$ according to Definition~\ref{def:index}.\\
Find $s'_{z-j}='('$ in $S$ with index $(i\mod (t+1))$ and output $A'=A\cup \{z\}$.
\end{algorithm}

\begin{theorem}\label{thm:alg}
\begin{enumerate}
\item Given a canonical $S'$, we can compute the indices of all parentheses in $S'$ in $O(t)$ time. Therefore, Algorithm~\ref{alg:flexP1} solves P1 in $O(t)$ time.
\item An instance $(A',i)$ of P2 reduces to the instance $([v]-A',j)$ of P1, where $i\mod (t+1)+j\mod(t+1)=t$. Thus P2 can be solved in $O(t)$ time.
\end{enumerate}
\end{theorem}

The proof of Theorem~\ref{thm:alg} is trivial and is omitted due to space limits.

\subsection{Variation laws of $\flex$}

\newcommand{\rhoRT}{\rho^{\times\rightarrowtail\triangle}}
\newcommand{\rhoRC}{\rho^{\times\rightarrowtail\fullmoon}}
\newcommand{\rhoLT}{\rho^{\triangle\leftarrowtail\times}}
\newcommand{\rhoLC}{\rho^{\fullmoon\leftarrowtail\times}}

We prove some \emph{variations laws} of $\flex$ as summarized in Lemma~\ref{lemma:flex-law},
which are comparable to the laws of modular factorization given below in Lemma~\ref{lemma:modular-law}.
See \ref{subsect:flex-law} for the details, including the definitions of
  $\rhoRT$, $\rhoRC$, $\rhoLC$, and $\rhoLT$.

\begin{lemma}[Variation laws of $\flex$]\label{lemma:flex-law}Restrict the remainder to $[t+1]$ here.\smallskip

\noindent When $\flex(\rho) \neq t+1$, $\flex(\rhoRT)=\flex(\rhoLC)=(\flex(\rho)-1) \mod (t+1)$.\smallskip

\noindent When $\flex(\rho) \neq t$, $\flex(\rhoRC)=\flex(\rhoLT)=(\flex(\rho)+1) \mod (t+1)$.
\end{lemma}


\section{Revisit the modular factorization}\label{sect:modular}

This section presents a new and simpler definition of the modular factorization.
When a number modulo $t+1$ in this section, the remainder is restricted to $[t+1]$.

\medskip \noindent \textbf{\underline{The modular factorization}\cite{MODULAR94}.\mbox{  }}
The modular factorization was originally given by $t+1$ 1-factors $\calM_1,\ldots,\calM_{t+1}$ where $\calM_i$ was defined as follows.
Consider $A\in \calP_t$. Let $\Sigma A$ indicate the sum of elements in $A$.
Let $y=(\Sigma A + i) \mod (t+1)(\in[t+1])$.
Then, $\calM_i(A):=A\cup \{z\}$, where $z$ is the $y$-th \textbf{largest} element in $[v]-A$.

Take $t=3,v=7$, and $A=\{2,4,6\}$ for example:

For $i=1$, we have $y=13=1\hspace{-2mm}\pmod 4$ and $z=7$. So $\calM_1(A)=\{2,4,6,7\}$.

For $i=2$, we have $y=14=2\hspace{-2mm}\pmod 4$ and $z=5$. So $\calM_2(A)=\{2,4,5,6\}$.

For $i=3$, we have $y=15=3\hspace{-2mm}\pmod 4$ and $z=3$. So $\calM_3(A)=\{2,3,4,6\}$.

For $i=4$, we have $y=16=4\hspace{-2mm}\pmod 4$ and $z=1$. So $\calM_4(A)=\{1,2,4,6\}$.

\smallskip \noindent\textbf{Note 1.} It is proved in \cite{MODULAR94} that $\calM_i$ is a 1-factor for each $i~(1\leq i\leq t+1)$. Moreover, it is obvious that all the 1-factors $\calM_1,\ldots,\calM_{t+1}$ are pairwise-disjoint.

\smallskip \noindent\textbf{Note 2.} The origins of modular factorization are murky, said by the authors of \cite{MODULAR94}, who credited it to Robinson, who asked if it is the same as the lexical one.

\smallskip \noindent\textbf{Note 3.} Assume $\calM_i(A)=A'$. We can compute $A$ from $i$ and $A'$ symmetrically.
    Let $x=(\Sigma A'+i) \mod (t+1) (\in [t+1])$ where $\Sigma A'$ indicates the sum of elements in $A'$. Then $A= A'-\{z\}$, where $z$ is the $x$-th \textbf{smallest} element in $A'$ \cite{MODULAR94}.
    Thus, the problems on modular factorizations analogous to P1 and P2 are easy to solve.

\medskip The original definition of the modular factorization above does not explicitly give its labeling function.
Such a labeling function will be needed in analyzing the variation laws of the above modular factorization in Lemma~\ref{lemma:modular-law} below and
  hence we state it Lemma~\ref{lemma:modular-lf}.
However, our definition of the modular factorization is \textbf{not} given by Lemma~\ref{lemma:modular-lf}.
The proof of Lemma~\ref{lemma:modular-lf} can be found in \ref{sect:modular-proof}.

\newcommand{\rankCCWT}{\mathsf{rank}^{\circlearrowleft}_{\triangle}}
\newcommand{\rankCCWC}{\mathsf{rank}^{\circlearrowleft}_{\fullmoon}}

\smallskip Consider any permutation $\rho=(\rho_1,\ldots,\rho_{2t+1})$ of $[t \fullmoon, t \triangle, 1 \times]$.
For each $i\in [2t+1]$, the \emph{position} of $\rho_i$ is $i$.
Let $O^\rho_1,\ldots,O^\rho_t$ be the positions of $t$ $\fullmoon$'s in $\rho$ and $T^\rho_1,\ldots,T^\rho_t$ the positions $t$ $\triangle$'s.
Denote by $\rankCCWT(\rho)$ the rank of $\times$ when enumerating all $\triangle$'s and $\times$ in $\rho$ from $\rho_{2t+1}$ back to $\rho_1$.
So, $\rankCCWT(\rho)-1$ is the number of $\triangle$'s with positions larger than the position of $\times$.
Denote by $\rankCCWC(\rho)$ the rank of $\times$ when enumerating all $\fullmoon$'s and $\times$ in $\rho$ from $\rho_{2t+1}$ back to $\rho_1$.

\begin{lemma}\label{lemma:modular-lf}
The labeling function of $\{\calM_1,\ldots,\calM_{t+1}\}$ is given by $\fm$, where
\begin{eqnarray*}
\fm(\rho)&:=&\rankCCWT(\rho)-\Sigma_{j=1}^t O^\rho_j \hspace{-1mm}\pmod{t+1} (\in [t+1]), \hbox{ or }\label{def:fm}\\
\fm(\rho)&:=&1 + \Sigma_{j=1}^t T^\rho_j - \rankCCWC(\rho) \hspace{-1mm}\pmod{t+1} (\in [t+1]). \label{def:fm-equiv}
\end{eqnarray*}
\end{lemma}

We now introduce a labeling function $\fmod$ and proves that $\fmod\equiv \fm+C$ for some constant $C$. Thus we give an alternative yet equivalent definition of the modular factorization, which is $\{F_{\fmod,1},\ldots,F_{\fmod,t+1}\}$.

\begin{definition}\label{def:fmod}
Assume $\rho=(\rho_1,\ldots,\rho_{2t+1})$ is any permutation of $[t \fullmoon, t \triangle, 1 \times]$.
Arrange $\rho_1,\ldots,\rho_{2t+1}$ in CW order.
We count \textbf{the number of tuples $(\times, \fullmoon, \triangle)$} which are located in CW order within this cycle of characters (positions may be inconsecutive) (such a tuple is an inversion when we cut the sequence at $\times$).
Taken modulo $(t+1)$, the remainder, restricted to $[t+1]$, is $\fmod(\rho)$. See Fig.~\ref{fig:fmod}.
\end{definition}

By Definition~\ref{def:fmod}, we establish an interesting connection between the modular factorization and
  the \textbf{inversion number of permutations} (section 5.3 of \cite{book-KT}).

\newcommand{\rhoCT}{\rho^{\times\rightarrow\triangle}}
\newcommand{\rhoCC}{\rho^{\times\rightarrow\fullmoon}}
\newcommand{\rhoDT}{\rho^{\triangle\leftarrow\times}}
\newcommand{\rhoDC}{\rho^{\fullmoon\leftarrow\times}}

\begin{figure}[t]
  \centering
  \includegraphics[width=.6\textwidth]{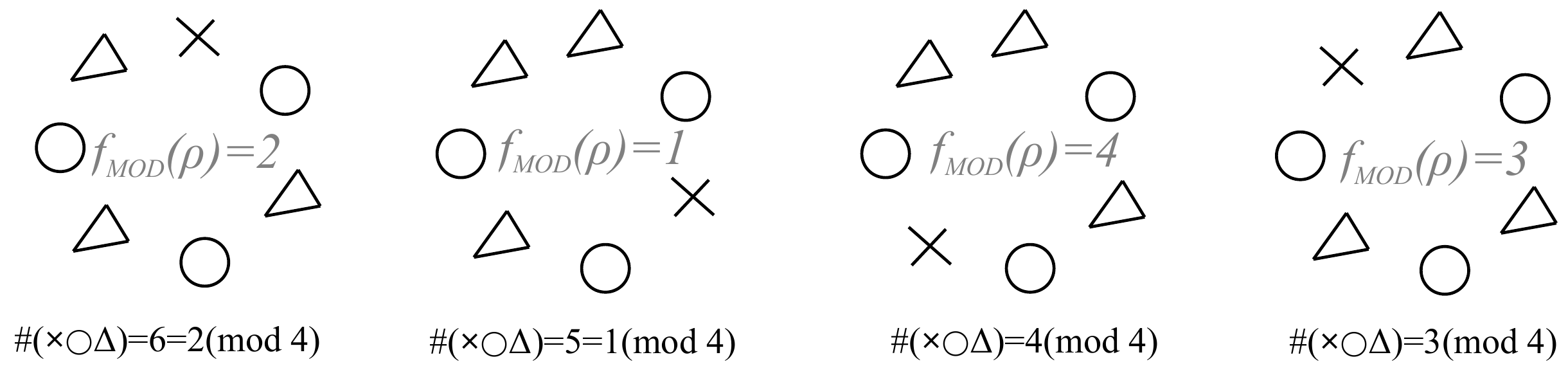}
  \caption{Illustration of the definition of $\fmod$. The four permutations drawn here share the same positions of $\fullmoon$'s, and they are mapped to different numbers under $\fmod$.} \label{fig:fmod}
\end{figure}

\medskip Let $\rhoCT$ be constructed from $\rho$, which swaps $\times$ with its CW next $\triangle$.

Let $\rhoCC$ be constructed from $\rho$, which swaps $\times$ with its CW next $\fullmoon$.

Let $\rhoDT$ be constructed from $\rho$, which swaps $\times$ with its CCW next $\triangle$.

Let $\rhoDC$ be constructed from $\rho$, which swaps $\times$ with its CCW next $\fullmoon$.

\begin{lemma}[Variation laws of $\fm$ and $\fmod$]\label{lemma:modular-law}
\begin{eqnarray}
\fmod(\rhoCT) = \fmod(\rhoDC) = \fmod(\rho)-1 \quad  (\hspace{-1mm}\mod{t+1}),\label{eqn:fmod-1}\\
\fmod(\rhoCC) = \fmod(\rhoDT) = \fmod(\rho)+1 \quad  (\hspace{-1mm}\mod{t+1}).\label{eqn:fmod+1}\\
\fm(\rhoCT) = \fm(\rhoDC) = \fm(\rho)-1,     \quad  (\hspace{-1mm}\mod{t+1}) \label{eqn:fm-1}\\
\fm(\rhoCC) = \fm(\rhoDT) = \fm(\rho)+1.    \quad  (\hspace{-1mm}\mod{t+1}) \label{eqn:fm+1}.
\end{eqnarray}
\end{lemma}

Lemma~\ref{lemma:modular-law} is proved in \ref{sect:modular-proof}. Its corollary below is trivial; proof omitted.

\begin{corollary}\label{corol:fm-fmod-equiv}
Because $\fm$ and $\fmod$ have the same variation law, there is a constant $C$ so that $\fmod\equiv \fm+C$.
Specifically, $\left\{
  \begin{array}{ll}
    C=0, & ~\hbox{$t$ is even;} \\
    C=(t+1)/2, & ~\hbox{$t$ is odd.}
  \end{array}
\right.$
\end{corollary}

At last, we point out that the resolved functions of $\fm$ or $\fmod$ can easily be deduced according to the original definition of modular factorization.
\bibliographystyle{splncs04}
\bibliography{1fac}

\clearpage

\appendix

\section{Proofs omitted in section~\ref{sect:resolvable-construct}}\label{sect:resolvable-proof}

This appendix contains the proofs of Lemmas~\ref{lemma:some-resolvable-cases}, \ref{lemma:regardlessofcolor}, and \ref{lemma:PA-CPA}.

\subsubsection{Restatement of Lemma~\ref{lemma:some-resolvable-cases}.}
\emph{\begin{enumerate}
\item Any 1-factorization of $H(v=4,t=1)$ is resolvable.
\item Any 1-factorization of $H(2t+1,t)$, including the lexical factorization and modular factorization (as illustrated in Fig.~\ref{fig:g_A_73}), is resolvable.
\item No 1-factorization of $H(6,2)$ is resolvable.
\end{enumerate}}

\begin{figure}[h]
  \centering
  \includegraphics[width=.65\textwidth]{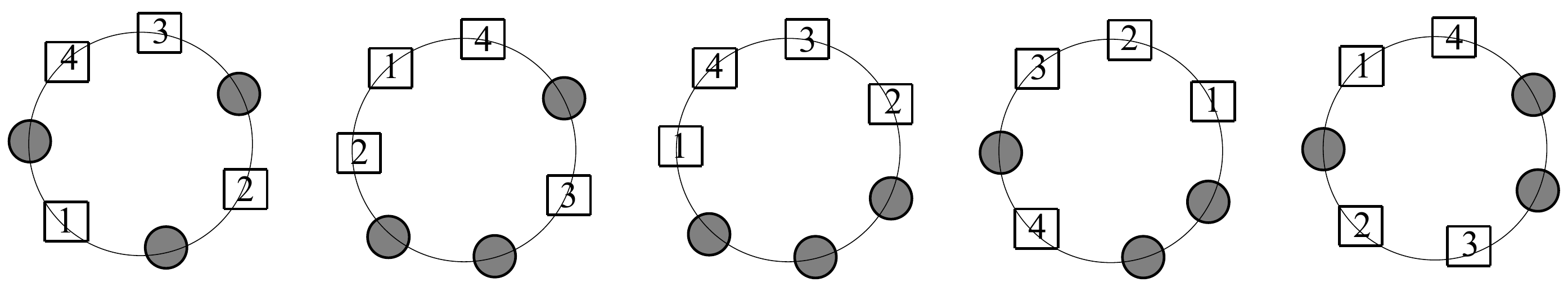}\\ \bigskip
  \includegraphics[width=.65\textwidth]{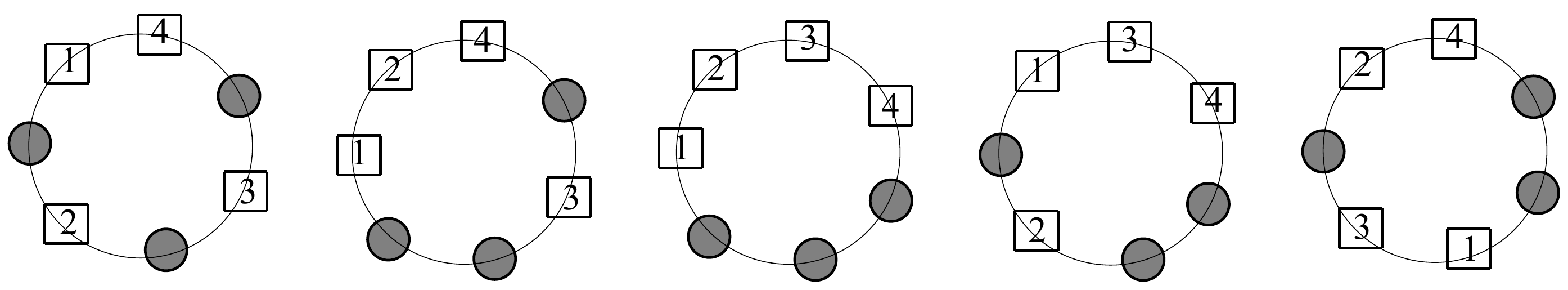}\\
  \caption{Top (bottom) shows resolvable functions for $\fmod$ ($\flex$) for $v=7,t=3$.}\label{fig:g_A_73}
\end{figure}

\begin{proof}[of claim 1 of Lemma~\ref{lemma:some-resolvable-cases}]
Consider any 1-factorization of $H(4,1)$.
Assume its labeling function is $f$.
We shall find $\{g_A \mid A\in \calP_1\}$ and $\gamma$ so that $f=f_{\gamma,g}$.

First, choose $\gamma$ to be any bijection from the $2$-element subset of $[3]$ to $[3]$.
Then, define $g_{\{1\}}$ as follows and define $g_{\{2\}}$, $g_{\{3\}}$, $g_{\{4\}}$ using a similar idea.

Observe that there must exist distinct numbers $a,b,c$ so that $f(\{1\},\{1,2,3\})=\gamma(\{a,b\})$
and $f(\{1\},\{1,2,4\})=\gamma(\{a,c\})$ and $f(\{1\},\{1,3,4\})=\gamma(\{b,c\}$.
This is because $\{f(\{1\},\{1,2,3\}),f(\{1\},\{1,2,4\}),f(\{1\},\{1,3,4\})\}=\{1,2,3\}$ whereas
  the preimages of $1,2,3$ under $\gamma$ are the three 2-element subsets of $[3]$.

By choosing $a,b,c$ as the values of $g_{\{1\}}(2),g_{\{1\}}(3),g_{\{1\}}(4)$ respectively,
\begin{eqnarray*}
f_{\gamma,g}(\{1\},\{1,2,3\}) = \gamma(g_{\{1\}}(\{2,3\})) = \gamma(\{a,b\}) & = f(\{1\},\{1,2,3\})\\
f_{\gamma,g}(\{1\},\{1,2,4\}) = \gamma(g_{\{1\}}(\{2,4\})) = \gamma(\{a,c\}) & = f(\{1\},\{1,2,4\})\\
f_{\gamma,g}(\{1\},\{1,3,4\}) = \gamma(g_{\{1\}}(\{3,4\})) = \gamma(\{b,c\}) & = f(\{1\},\{1,3,4\})
\end{eqnarray*}
Therefore, it is always possible to find a group of $\{g_A\}$ such that $f=f_{\gamma,g}$. \qed
\end{proof}

\begin{proof}[of claim 2 of Lemma~\ref{lemma:some-resolvable-cases}]
Consider any 1-factorization of $H(2t+1,t)$.
Assume its labeling function is $f$.
Choose $\gamma$ to be the natural bijection from all the $1$-element sets of $[t+1]$ to $[t+1]$, which maps $\{x\}$ to $x$.
We shall find $\{g_A\}$ so that $f=f_{\gamma,g}$.
For any $t$-subset $A$ of $[2t+1]$, we define $g_A$ as follows.
\begin{equation}
g_A(x):=f(A,A\cup \{x\})\quad (\forall x\in A^C).
\end{equation}

Because $f$ defines a 1-factorization, it satisfies condition (a). Therefore, when $x$ is taken over all elements in $A^C$, function
  $f(A,A\cup \{x\})$ would be taken over all numbers in $[t+1]$.  This means that $g_A$ is indeed a bijection from $A^C$ to $[t+1]$.
  Moreover, it is straightforward to see $f_{\gamma,g}(A,A \cup \{x\})\equiv f(A,A\cup      \{x\})$. \qed
\end{proof}

\begin{proof}[of claim 3 of Lemma~\ref{lemma:some-resolvable-cases}]
We design a short \emph{C++} program which searches all the resolvable 1-factorizations of $H(6,2)$ and $H(8,3)$ by brute force,
which can be downloaded at \url{https://github.com/cscjjk/resolvable-1-factorization}.

\smallskip For $H(8,3)$, the program returns many solutions.

For $H(6,2)$, the program runs in less than five seconds and finds no solution.
This shows that there is no resolvable 1-factorization of $H(6,2)$.

\smallskip We note that this claim is not so important for this manuscript because no result is depending on this claim. So we only prove it by a \emph{C++} program. \qed
\end{proof}

\subsubsection{Restatement of Lemma~\ref{lemma:regardlessofcolor}.}
\emph{Given $\{g_A\}$ and $\gamma$, the following are equivalent:
\begin{enumerate}
\item[(i)] function $f_{\gamma,g}(A,A')$ satisfies condition (b); and
\item[(ii)] When $A_1\neq A_2$ and $(A_1,A'_1),(A_2,A'_2)$ are two edges in $H(v,t)$, then $g_{A_1}(A'_1-A_1)= g_{A_2}(A'_2-A_2)$ implies that $A'_1\neq A'_2$.
\end{enumerate}}

\begin{proof}[of Lemma~\ref{lemma:regardlessofcolor}]
Assume (i) holds. Proving (ii) is equivalent to proving that
    for $A_1,A_2,A'_1,A'_2$ such that $A_1\neq A_2$ and $(A_1,A'_1),(A_2,A'_2)$ are two edges in $H(v,t)$,
  equality $A'_1=A'_2=A'$ implies that $g_{A_1}(A'_1-A_1)\neq g_{A_2}(A'_2-A_2)$.

Because condition (b) is satisfied by $f_{\gamma,g}$ whereas $A_1\neq A_2$ and $A'_1=A'_2=A'$, we get
$f_{\gamma,g}(A_1,A')\neq f_{\gamma,g}(A_2,A')$, i.e., $\gamma (g_{A_1}(A'-A_1))\neq \gamma (g_{A_2}(A'-A_2)).$
Further since $\gamma$ is a bijection, $g_{A_1}(A'_1-A_1) \neq g_{A_2}(A'_2-A_2)$. \medskip

Now we prove (i) from (ii). Assume (ii) holds. For any two distinct edges $(A_1,A')$ and $(A_2,A')$ of $H(v,t)$,
  we shall prove that $f_{\gamma,g}(A_1,A')\neq f_{\gamma,g}(A_2,A')$, namely, $g_{A_1}(A'-A_1)\neq g_{A_2}(A'-A_2)$.
Suppose to the opposite that $g_{A_1}(A'-A_1)=g_{A_2}(A'-A_2)$, we get $A'\neq A'$ according to (ii), which is contradictory. \qed
\end{proof}

\subsubsection{Restatement of Lemma~\ref{lemma:PA-CPA}.}
\emph{Any $t+d$ columns of a $\PA(t,2t+d,2t+d)$ form a $\CPA(t,t+d,2t+d)$.}

\begin{proof}[of Lemma~\ref{lemma:PA-CPA}]
Suppose $M'$ contains any $t+d$ columns of a $\PA(t,2t+d,2t+d)$ $M$.
By the definition of perpendicular arrays, $M'$ is still a $\PA(t,t+d,2t+d)$.
Therefore, we only need to show that $M'$ is complete.

Let $A_i$ denote the set of elements in the $i$-th row of $M$ among those $t$ columns which are not chosen in constructing $M'$.
Because $M$ is $\PA(t,2t+d,2t+d)$, sets $A_1,\ldots, A_{2t+d \choose t}$ go through every $t$-element subset of $[2t+d]$ exactly once.
Therefore, $[2t+d]-A_1,\ldots,[2t+d]-A_{2t+d \choose t}$ go through every $(t+d)$-element subset of $[2t+d]$ exactly once.
However, the elements in each row of $M'$ are respectively $[2t+d]-A_1,\ldots,[2t+d]-A_{2t+d \choose t}$.
Together, $M'$ is complete. \qed
\end{proof}

\clearpage

\section{Proofs omitted in section~\ref{sect:lexical}}\label{sect:lexical-proof}

\subsubsection{Restatement of Lemma~\ref{lemma:cyclic-unique}.}
\emph{Given any sequence $S$ of $t$ `)'s and $t+1$ `('s.
There exists a unique cyclic-shift $S^{(j)}$ of $S$ whose first $2t$ parentheses are paired up when parenthesized,
and we can compute $j$ in $O(t)$ time.}

\begin{proof}[of Lemma~\ref{lemma:cyclic-unique}]
Assume $S$ is a sequence of $v=2t+1$ parentheses, $t$ of which are `)'s.
We are interested in finding a cyclic-shift of $S$ in which the first $2t$ parentheses can be paired up when parenthesizing.

Denote $H_i =\hbox{the number of `)'s} - \hbox{the number of `('s in }s_1,\ldots,s_i$ for each $i$.
Draw points $\{(i,H_i)\mid 0\leq i\leq v\}$ in the Cartesian plane, as shown in Fig.~\ref{fig:cyclicunique}.

Select the highest point $(j^*,H_{j^*})$; for a tie, select the rightmost one.

When $j\neq j^*+1$, the cyclic-shift $S^{(j)}$ does not satisfy our requirement.
This is because when $j\neq j^*+1$, the one shifted from $s_{j^*+1}$, which is a left parenthesis, cannot be paired up.
When $j=j^*+1$, the cyclic-shift $S^{(j)}$ satisfies our requirement. This is simply illustrated in the figure.
To complete, we point out that index $j^*+1$ can easily be computed in $O(t)$ time.

This lemma also follows from Cycle Lemma \cite{CycleLemma} or Raney Lemma \cite{RaneyLemma}.  \qed
\end{proof}

\begin{figure}[b]
\begin{minipage}[b]{.55\textwidth}
  \centering
  \includegraphics[width=.9\textwidth]{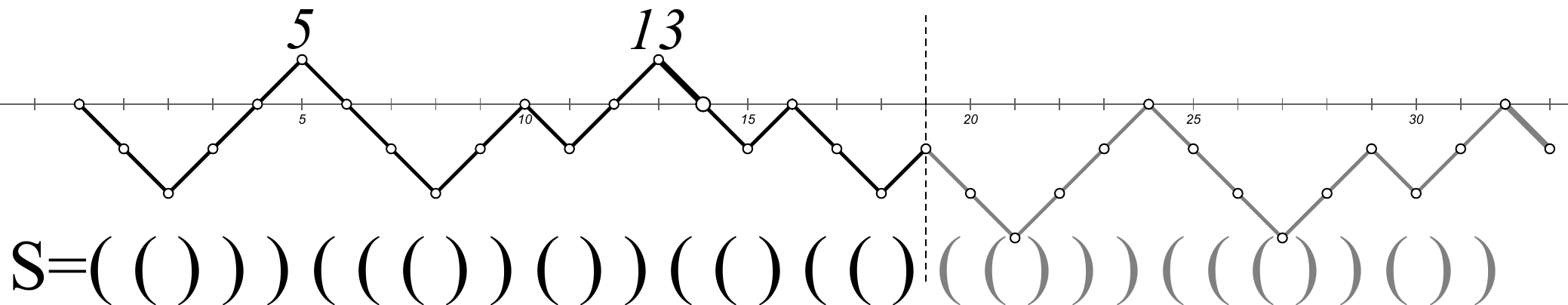}\\
  \caption{Illustration of Lemma~\ref{lemma:cyclic-unique}.\newline This figure draws Example~\ref{example:S}.}\label{fig:cyclicunique}
\end{minipage}
\begin{minipage}[b]{.45\textwidth}
\centering \includegraphics[width=.8\textwidth]{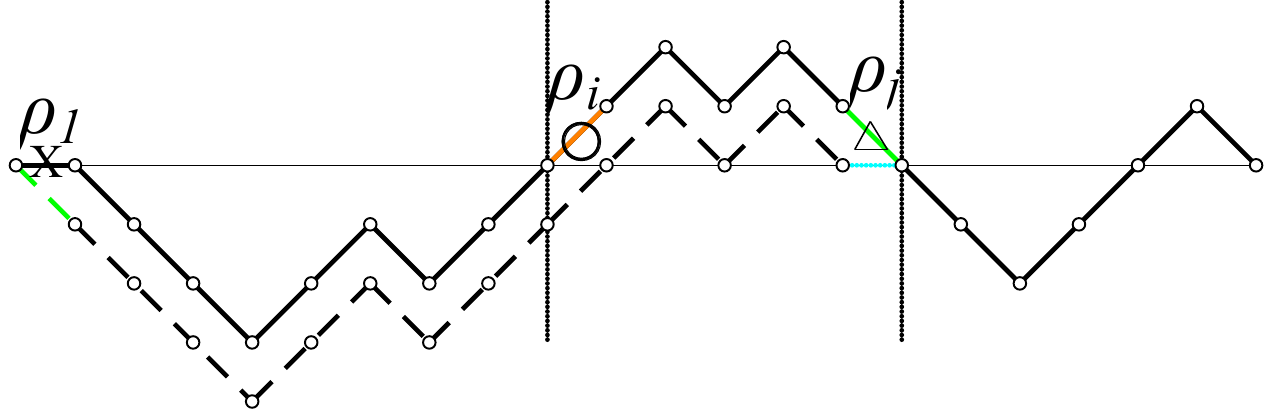}\\
  \caption{Illustration of variation laws below. The dotted line indicates $\rhoRT$.}\label{fig:positivelaw}
\end{minipage}
\end{figure}

\subsubsection{Restatement of Lemma~\ref{lemma:index-stronger}.}
\emph{When $S$ is canonical, for any $s_l=($ and $s_r=)$,
  there are more $)$'s than $($'s in $\{s_{l+1},\ldots,s_{r}\}$ $\Leftrightarrow$ $\ind(s_l)\geq \ind(s_r)$.}

\begin{proof}[of Lemma~\ref{lemma:index-stronger}] For each $i~(0\leq i\leq 2t+1)$, denote

\quad $H_i =\hbox{the number of `)'s} - \hbox{the number of `('s in }s_1,\ldots,s_i$, and

\quad $H'_i=\hbox{the number of `)'s} - \hbox{the number of `('s in }s_{i+1},\ldots,s_{2t+1}$.

To be clear, $H_0=H'_{2t+1}=0$. Trivially, we have:

$H_i+H'_i=-1 ~(\hbox{for all }i).\qquad \dep(s_l)=-H_l-1. \qquad \dep(s_r)=-H_r.$

We need to discuss two cases depending on which one of $l,r$ is smaller.
\begin{itemize}
\item Case~1: $l<r$. Let us count the number of ')'s  minus the number of '('s in the interval $s_{l+1},\ldots,s_r$. This is given by $H_r-H_l$. Therefore,
there are more ')'s in this interval $\Leftrightarrow H_r-H_l>0 \Leftrightarrow -\dep(s_r)+\dep(s_l)+1>0  \Leftrightarrow \dep(s_r)\leq \dep(s_l) \Leftrightarrow \ind(s_r) \leq \ind(s_l) ~\hbox{(when $l<r$)}$. \smallskip

\item Case~2: $r<l$. The number of ')' minus the number of '('s in (cyclic) interval $s_{l+1},\ldots,s_{2t+1},s_1,\ldots,s_r$ is given by $H'_l+H_r=-1-H_l+H_r$.  So,
there are more ')'s in this interval $\Leftrightarrow -1-H_l+H_r>0 \Leftrightarrow \dep(s_l)-\dep(s_r)>0
\Leftrightarrow \dep(s_r)< \dep(s_l)  \Leftrightarrow \ind(s_r)\leq \ind(s_l) ~\hbox{(when $r<l$)}.$

\end{itemize}
In either case, the lemma holds. \qed
\end{proof}


\clearpage

\subsection{The full version of the variation laws of $\flex$}\label{subsect:flex-law}

Consider any permutation $\rho$ of $[t \fullmoon, t \triangle, 1 \times]$. For any character $\triangle$ or $\fullmoon$ in $\rho$,

we say it is \emph{CW-balanced} if there are equal number of $\triangle$'s and $\fullmoon$'s
  in the (cyclic) interval of $\rho$ starting from $\times$ to this character in CW order, and

we say it is \emph{CCW-balanced} if there are equal number of $\triangle$'s and $\fullmoon$'s
  in the (cyclic) interval of $\rho$ starting from $\times$ to this character in CCW order.

\smallskip The following lemma is the full version of Lemma~\ref{lemma:flex-law}.

\begin{lemma}[Variation laws of $\flex$]\label{lemma:flex-law-pre}
\begin{enumerate}
\item $\flex(\rho)\neq t+1$ $\Leftrightarrow$ there is a CW-balanced $\triangle$ $\Leftrightarrow$ there is a CCW-balanced $\fullmoon$.\smallskip
\item $\flex(\rho)\neq t$ $\Leftrightarrow$ there is a CW-balanced $\fullmoon$ $\Leftrightarrow$  there is a CCW-balanced $\triangle$.\smallskip
\item When $\flex(\rho)\neq t+1$, let $\rhoRT$ ($\rhoLC$) be constructed from $\rho$ by swapping $\times$ with the CW first CW-balanced $\triangle$ (the CCW first CCW-balanced $\fullmoon$).\\ Then, $$\flex(\rhoRT)=\flex(\rhoLC)=(\flex(\rho)-1) \mod (t+1)(\in [t+1]).$$
\item When $\flex(\rho)\neq t$, let $\rhoRC$ ($\rhoLT$) be constructed from $\rho$ by swapping $\times$ with the CW first CW-balanced $\fullmoon$ (the CCW first CCW-balanced $\triangle$).\\ Then, $$\flex(\rhoRC)=\flex(\rhoLT)=(\flex(\rho)+1) \mod (t+1)(\in [t+1]).$$
\end{enumerate}
\end{lemma}

\begin{proof}[of Lemma~\ref{lemma:flex-law-pre}] Without loss of generality, assume that $\rho_1=\times$.
For each $i~(1\leq i\leq v)$, define the \emph{height} of $\rho_i$ as the number of $\fullmoon$'s minus the number of $\triangle$'s in $\{\rho_1,\ldots,\rho_i\}$.
(So, a $\fullmoon$ is positive if and only if its height is positive.)

\medskip \noindent \emph{Proof of Claim~1.} Assume $\flex(\rho)\neq t+1$. In this case there exists some $\fullmoon$ with positive height.
This implies that there exists a pair of $(i,j)$ such that $\rho_i=\fullmoon$ has a height $1$ while $\rho_j=\triangle$ has a height $0$,
as shown in Fig.~\ref{fig:positivelaw}. Clearly, $\rho_j$ is a CW-balanced $\triangle$ while $\rho_i$ is a CCW-balanced $\fullmoon$.
On the other direction, the existence of a CW-balanced $\triangle$ or a CCW-balanced $\fullmoon$ implies the existence of a positive $\fullmoon$, which immediately implies that $\flex(\rho)\neq t+1$.

Claim~2 is symmetric to Claim~1; proof omitted. \qed

\medskip \noindent \emph{Proof of the equations in Claim~3 and Claim~4.} Because the four equations are symmetric, we only show the proof of $\flex(\rhoRT)=\flex(\rho)-1$.

Without loss of generality, assume $\rho_1=\times$.
Let $\rho_i$ be the CW first $\fullmoon$ with height $1$.
Let $\rho_j$ be the CW first $\triangle$ with height $0$, i.e.\ the CW first CW-balanced $\triangle$.
As illustrated in Fig.~\ref{fig:positivelaw}, $\rhoRT$ is constructed from $\rho$ by swapping $\rho_1$ with $\rho_j$.
We shall prove that after the swapping, the number of positive $\fullmoon$'s decreases by 1.
This follows from three observations:
(i) $\rho_i=\fullmoon$ is positive in $\rho$ (with height 1) but not anymore in $\rhoRT$ (with height 0).
(ii) For other $\fullmoon$'s in $\rho_2,\ldots,\rho_j$, their heights drop by 1, but their positivity do not change.
(iii) For the $\fullmoon$'s in $\rho_{j+1},\ldots,\rho_{2t+1}$, their heights and positivity stay the same as before.\qed
\end{proof}


\clearpage

\section{Proofs omitted in section~\ref{sect:modular}}\label{sect:modular-proof}

\subsubsection{Restatement of Lemma~\ref{lemma:modular-lf}.}
\emph{The labeling function of $\{\calM_i\mid 1\leq i\leq t+1\}$ is}
\begin{eqnarray}
\fm(\rho)&:=&\rankCCWT(\rho)-\Sigma_{j=1}^t O^\rho_j \hspace{-1mm}\pmod{t+1} (\in [t+1]), \hbox{ or }\label{def:fm}\\
\fm(\rho)&:=&1 + \Sigma_{j=1}^t T^\rho_j - \rankCCWC(\rho) \hspace{-1mm}\pmod{t+1} (\in [t+1]). \label{def:fm-equiv}
\end{eqnarray}

\begin{proof}[of Lemma~\ref{lemma:modular-lf}]
We first state two trivial observations:

$(\times\hbox{'s position}) + \Sigma_j O^\rho_j + \Sigma_j T^\rho_j  = 1+\ldots+(2t+1) = 0 (\mod{t+1})$, and

$(\times\hbox{'s position}) + \rankCCWT(\rho)-1 + \rankCCWC(\rho)-1=2t+1=-1 (\mod{t+1})$.

By subtraction, $\rankCCWT(\rho)-\Sigma_{j=1}^t O^\rho_j = 1 + \Sigma_{j=1}^t T^\rho_j - \rankCCWC(\rho) (\mod{t+1})$.
Therefore, the two definitions of $\fm$ given in (\ref{def:fm}) and (\ref{def:fm-equiv}) are equivalent.

Next, we show that $\fm$ is the labelling function of $\{\calM_1,\ldots,\calM_{t+1}\}$.
Recall that $\rho$ represents the edge $(A,A')$ in the middle level graph, where
  $A=\{O^\rho_1,\ldots,O^\rho_t\}$ and $A'=\{O^\rho_1,\ldots,O^\rho_t,\hbox{the position of }\times\}$.
We shall prove that $(A,A')\in \calM_{\fm(\rho)}$.
By the definition of $\calM_{\fm(\rho)}$, it reduces to proving that the single element in $A'-A$ is the $y$-th largest one in $[v]-A$, where $y=(\Sigma A+\fm(\rho)) \mod (t+1)$ ($y\in [t+1]$).
Namely, the unique $\times$ has rank $y$ when enumerating all $\triangle$'s or $\times$ in $\rho$ in CCW;
i.e., $\rankCCWT(\rho)=y\mod (t+1)$.
This holds since $y=\Sigma A+\fm(\rho)=\Sigma_{j=1}^t O^\rho_j + \rankCCWT(\rho)-\Sigma_{j=1}^t O^\rho_j\mod (t+1)$. \qed
\end{proof}

\smallskip Next, recall the variation laws of $\fm$ and $\fmod$ in Lemma~\ref{lemma:modular-law}.
Notice that (\ref{eqn:fmod+1}) and (\ref{eqn:fm+1}) are equivalent to (\ref{eqn:fmod-1}) and (\ref{eqn:fm-1}) respectively.
We now prove (\ref{eqn:fmod-1}) and (\ref{eqn:fm-1}).

\begin{proof}[of (\ref{eqn:fmod-1})]
We only need to prove $\fmod(\rhoCT)=\fmod(\rho)-1 (\mod{t+1})$.
The other equation $\fmod(\rhoDC)=\fmod(\rho)-1 (\mod{t+1})$ in (\ref{eqn:fmod-1}) is symmetric.

See Fig.~\ref{fig:countproof}. Denote by $a$ the number of $\fullmoon$'s between $\times$ and its CW next $\triangle$ in $\rho$.
Recall that $\fmod(\rho)$ denotes the number of $(\times, \fullmoon, \triangle)$-tuples which are located in CW order within $\rho$ (and then modulo $t+1$).
So, $\fmod(\rhoCT)-\fmod(\rho) = (t-a)\cdot 1-a\cdot t \mod (t+1)$,
which implies that $\fmod(\rhoCT) = \fmod(\rho)-1 \mod (t+1)$.
In calculating the difference between $\fmod(\rhoCT)$ and $\fmod(\rho)$, observe that
(i) for the $\fullmoon$'s located CW between the $\triangle$ being swapped and $\times$,
  the number of $(\times, \fullmoon, \triangle)$-tuples related to each of them increases by $1$; and
(ii) for the other $\fullmoon$'s, the number of $(\times, \fullmoon, \triangle)$-tuples related to each of them decreases by $t$. \qed
\end{proof}

\begin{figure}[h]
  \centering
  \includegraphics[width=.23\textwidth]{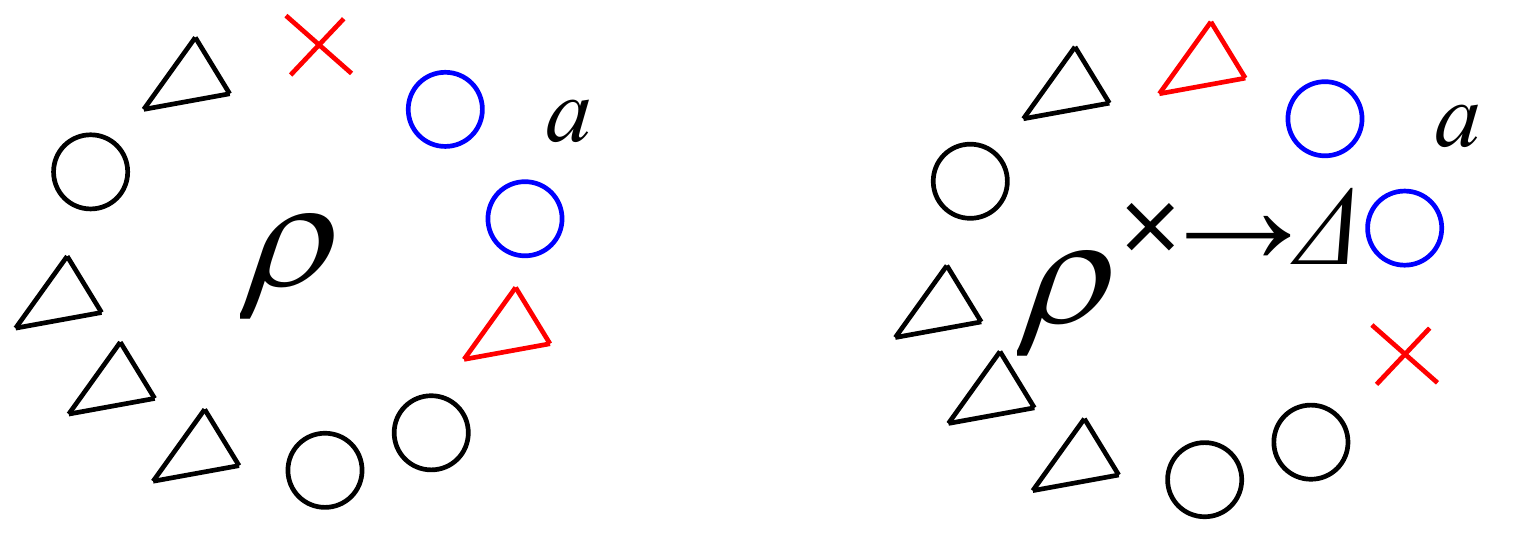}\\
  \caption{Illustration of the proof the variation law of $\fmod$ function..}\label{fig:countproof}
\end{figure}

\begin{proof}[of (\ref{eqn:fm-1})] By swapping $\times$ with its CW next $\triangle$, $\rankCCWT(\rho)$ decreases by 1.
Further by (\ref{def:fm}), $\fm(\rho)=\rankCCWT(\rho)-\Sigma_{j=1}^t O^\rho_j (\mod{t+1})$ decreases by 1.
Similarly, by swapping $\times$ with its CCW next $\fullmoon$, $\rankCCWC(\rho)$ increases by 1.
Further by (\ref{def:fm-equiv}), $\fm(\rho)=1 + \Sigma_{j=1}^t T^\rho_j - \rankCCWC(\rho)(\mod{t+1})$ decreases by 1. \qed
\end{proof}

\clearpage

\section{Explicit 1-factors of the bipartite Kneser graph}\label{sect:chain}

Although, most explicit 1-factorizations of the bipartite Kneser graph $H(v,t)$ have not been found, especially for $v>2t+1$,
  two explicit 1-factors of $H(v,t)$ are known for a long time.
To make the paper more self-contained, in this appendix we briefly review the literature of these 1-factors and give their new definitions.

\newcommand{\CCWRS}{\circlearrowleft}
\newcommand{\CWRS}{\circlearrowright}

\begin{definition}\label{def:1-factor} Assume $A\subset [v]$ and $|A| \leq v/2$.
By the following two steps, we can obtain a subset $A'$ which has equal size as $A$ and is disjoint with $A$, and we define it to be $\CCWRS(A)$.

\smallskip Step~1. Write down all the numbers in $[v]$ to a cycle from $1$ to $v$ in CW order.

\smallskip Step~2. Enumerate each number $a$ in $A$, find the CCW first number from $a$ that is not contained in $A\cup A'$ yet and add it to $A'$.

Note: The order of the enumeration in Step~2 does not matter.
Take $v=10$ and $A=\{1,3,8,9\}$ for example.
  In order $1,3,8,9$, the numbers added to $A'$ would be $10,2,7,6$.
    In order $3,9,8,1$, the numbers added to $A'$ would be $2,7,6,10$.)

\medskip We define $\CWRS(A)$ symmetrically (by changing CCW to CW in Step~2).\smallskip

\medskip Recall that $\calP_t$ denotes the $t$-th level of the subset lattice of $[v]$, i.e., it contains all subsets of $[v]$ with $t$ elements.
For $t<v/2$ and $A\subset \calP_t$, define
\begin{equation}\label{eqn:def-gamma-gamma'}
  \gamma_t^v(A):=[v]-\CCWRS(A)\hbox{ and }{\gamma'}_t^v(A):=[v]-\CWRS(A).
\end{equation}
\end{definition}

Obviously, $\gamma_t^v$ and ${\gamma'}_t^v$ are two 1-factors of $H(v,t)$, and they are disjoint.

\begin{lemma}
When $v=2t+1$, we have $\gamma_t^v=\calL_{t+1}$ and ${\gamma'}_t^v=\calL_{t}$.
\end{lemma}

\begin{proof}
We only show that $\gamma_t^v=\calL_{t+1}$. The other equation is similar.
Consider a subset $A\subset \calP_t$.
Replace all elements in $A$ by $\fullmoon$ and all the elements in $\CCWRS(A)$ by $\triangle$ and the remaining element by $\times$.
Clearly, this permutation is mapped to $t+1$ under $\flex$, because no $\fullmoon$ is positive.
This means $A$ is mapped to $[v]-\CCWRS(A)$ in $\calL_{t+1}$. Also, $A$ is mapped to $[v]-\CCWRS(A)$ in $\gamma_t^v$.  \qed
\end{proof}

In the following, we review a 1-factor $\beta_t^v$ of $H(v,t)$ and prove that $\beta_t^v=\alpha_t^v$.

\subsection{Definition of $\beta_t^v$ introduced in \cite{Bracket76}}

First, we review the chain-decomposition of the subset lattice given in \cite{Bracket76}.
Recall the parenthesis representation introduced in subsection~\ref{subsect:g_A-flex}.
The sequence of parentheses can be parenthesized uniquely in the usual way,
and there may remain several parenthesis unpaired. For example,
in $)_1,(_2,)_3,)_4,(_5,(_6,(_7,)_8,)_9,(_{10}$, ``$(_2$'' is paired with ``$)_3$'', ``$(_6$'' is paired with ``$)_9$'', and  ``$(_7$'' is paired with ``$)_8$''. All the others are unpaired.
Note that all the unpaired right parentheses always occur to the left of the unpaired left parentheses.\smallskip

\smallskip \noindent\textbf{\underline{Chain-decomposition of the subset lattice via parenthesizing}\cite{Bracket76}.}
Two subsets of $[v]$ are in the same chain, if and only if their associated parenthesis sequences contain the same paired parenthesis.
Equivalently, suppose $A\subset [v]$ is associated with sequence $S$.
Replace the leftmost unpaired '(' in $S$ by ')' and assume that the new sequence corresponds to subset $A'$. Then, $A'$ is the next member in the chain containing $A$.
  For the above example, the leftmost unpaired '(' is $(_5$, so $A'=\{1,3,4,5,8,9\}$.
  The entire chain in this example is $\{3,8,9\}\rightarrow \{1,3,8,9\}\rightarrow \{1,3,4,8,9\}\rightarrow \{1,3,4,5,8,9\}\rightarrow \{1,3,4,5,8,9,10\}.$

\smallskip Clearly, all chains in this decomposition are \emph{symmetric} -- if a chain contains a member $A$, it must contain a member with size $v-|A|$.
So, this chain-decomposition implicitly defines an antipodal matching $\beta_t^v$ between the antipodal layers $\calP_t$ and $\calP_{v-t}$ for each $t<v/2$.

\subsection{The equivalence between $\beta_t^v$ and $\gamma_t^v$.}

\begin{lemma}\label{lemma:gamma}
Assume $t<v/2$. We have $\beta_t^v(A)=\gamma_t^v(A)$ for any $A\in \calP_t$.
\end{lemma}

\newcommand{\PS}{\mathsf{PS}}
\begin{proof}
We shall prove that $\beta_t^v(A)=[v]-\CCWRS(A)$.
We first prove it by an example and then give the formal proof.
Let $\PS(A)$ denote the \emph{parenthesis sequence} associated with $A$.

\begin{example}
$v=11,A=\{1,3,4,8,9\}$. The sequence of parentheses associated with $A$ is:
$$\mbox{\qquad\qquad}\PS(A)=
\begin{array}{ccccccccccc}
\frame{)}_1&(_2&)_3&\frame{)}_4&\frame{(}_5&(_6&(_7&)_8&)_9&\frame{(}_{10}&\frame{(}_{11}.
\end{array}$$
The unpaired parentheses are boxed for ease of distinction.

There are two unmatched right parentheses and three unmatched left parentheses.
According to the definition of the chain-decomposition, in its symmetric member $\beta_5^{11}(A)$ we should replace the first unmatched left parenthesis by a right parenthesis.
So,
$$\mbox{\qquad}\PS(\beta_5^{11}(A))=
\begin{array}{ccccccccccc}
\frame{)}_1&(_2&)_3&\frame{)}_4&\frame{)}_5&(_6&(_7&)_8&)_9&\frame{(}_{10}&\frame{(}_{11}.
\end{array}$$

Then, let us also compute $\CCWRS(A)$ and $[v]-\CCWRS(A)$.
(In the following, the positions of boxes stay the same as above; they do not indicate the unpaired parentheses.)
$$\qquad \PS(\CCWRS(A))=
\begin{array}{ccccccccccc}
\frame{(}_1&)_2&(_3&\frame{(}_4&\frame{(}_5&)_6&)_7&(_8&(_9&\frame{)}_{10}&\frame{)}_{11}.
\end{array}$$
$$\PS([v]-\CCWRS(A))=
\begin{array}{ccccccccccc}
\frame{)}_1&(_2&)_3&\frame{)}_4&\frame{)}_5&(_6&(_7&)_8&)_9&\frame{(}_{10}&\frame{(}_{11}.
\end{array}$$

We see $\PS(\beta_5^{11}(A))=\PS([v]-\CCWRS(A))$. Therefore, $\beta_5^{11}(A)=[v]-\CCWRS(A)$.
\end{example}

For any $i\in [v]$, we shall prove that $\hbox{(X) $i\in \beta_t^v(A)$ if and only if $i\in [v]-\CCWRS(A)$}.$
We discuss two cases distinguished by whether $i$ belongs to $U$, where $U$ is the set of unpaired positions of $\PS(A)$ (the positions are indexed by $1,\ldots,v$).

\smallskip Case~1: $i\notin U$. Then, the $i$-th parenthesis of $\PS(A)$ is paired. It will not change within the chain containing $A$ and $\beta_t^v(A)$.
Therefore, (I) $i\in \beta_t^v(A)$ if and only if $i\in A$.
On the other hand, by the definition of $\CCWRS(A)$, it easily follows that $i\in \CCWRS(A)$ if and only if $i\notin A$.
(In the example above, the paired number 3 in $A$ will go to 2 in $\CCWRS(A)$, the paired numbers 8 and 9 will go to 6 and 7 in $\CCWRS(A)$. So $i\in \CCWRS(A)$ if and only if $i\notin A$.)
Therefore, (II) $i\in [v]-\CCWRS(A)$ if and only if $i\in A$. Combine (I) and (II), we get statement (X).\smallskip

\smallskip Case~2: $i\in U$. Assume $\PS(A)$ has $r$ unpaired right parentheses and $l$ unpaired left parentheses.
For any sequence $S$ with length $v$, let $S^{(U)}$ denote the subsequence of $S$ that are located at $U$.
We state the following arguments about the parentheses locating at $U$.

\begin{enumerate}
\item $\PS(A)^{(U)}$ starts by $r$ `)'s and is followed by $l$ `('s.
\item $\PS(\beta_t^v(A))^{(U)}$ \underline{starts by $l$ `)'s and is followed by $r$ `('s}.
\item $\PS(\CCWRS(A))^{(U)}$ starts by $l$ `('s and is followed by $r$ `)'s.
\item $\PS([v]-\CCWRS(A))^{(U)}$ \underline{starts by $l$ `)'s and is followed by $r$ `('s}.
\end{enumerate}

The first argument is according to the assumption of $l$ and $r$. The second follows by 1 and the fact that $\beta_t^v(A)$ is the symmetric member of $A$ in the chain containing them.
The third follows by 1 and the definition of the CCW-rotating-subset. The last follows by the third.
According to 2 and 4, we obtain (X) for those $i$ in $U$ altogether.  \qed
\end{proof}

\subsection{Remarks and related work}

\begin{remark}
According to Lemma~\ref{lemma:gamma}, our definition of $\gamma_t^v$ essentially gives an \textbf{explicit} definition of the antipodal matching $\beta_t^v$,
  which was previously defined implicitly from the chain-decomposition.

In fact, \cite{Antipodal-code} presented an even more explicit definition of $\beta_t^v$ using Cycle Lemma \cite{CycleLemma}.
Based on their definition, they further showed that $\beta_t^v(A)$ can be computed in $O(v)$ time and $O(\log v)$ space.
We do not review their work in depth in this appendix.
  (Note: we believe that \cite{Antipodal-code} in fact discusses the other 1-factor ${\gamma'}_t^v$ rather than $\gamma_t^v$, but it is straightforward to extend their result to the symmetric 1-factor $\gamma_t^v=\beta_t^v$.)
\end{remark}

\noindent \underline{\textbf{An equivalent definition of the chain-decomposition}}
A few years earlier than \cite{Bracket76}, Aigner \cite{Lexicographic} proposed a greedy algorithm which can produce a matching $\lambda_t$ between two consecutive layers $\calP_t,\calP_{t+1}$.
The $v$ matchings $\lambda_0$, \ldots, $\lambda_{v-1}$ together describe a chain-decomposition of the subset lattice.
Interestingly, \cite{Recur77} pointed out that this decomposition is the same as the above one introduced in \cite{Bracket76} via parenthesizing. This was not mentioned in \cite{Bracket76}.

\bigskip \noindent \underline{\textbf{Yet Another definition of the chain-decomposition.}}
  Recently, another alternative definition for the above chain-decomposition was proposed in \cite{anotherchain}. However, their definition looks extremely complicated.
  We do not introduce it in this manuscript.

\begin{remark}
The bipartite graph $H(v,t)$ admits two disjoint 1-factors according to the fact that it is Hamiltonian \cite{MS17-H}.
Recently, Spink \cite{chain3} found out three orthogonal chain decompositions of the subset lattice, which implies three disjoint 1-factors of $H(v,t)$.
More recently, four orthogonal chain decompositions can be found for $v\geq 60$ \cite{chain4}.
\end{remark}

\clearpage

\section{Application: unique-supply hat-guessing games}\label{sect:hatguessing-motivation}

Hat-guessing games have been studied extensively in a broad area due to their relations to graph entropy, circuit complexity, network coding, and auctions \cite{Hat-auc1,Hat-auc2,Hat2,Hat-hamming-ebert,Hat-new,Hat3}.
In this appendix we show applications of the 1-factorization of the bipartite Kneser graphs
  in the following variant of hat-guessing game:

\bigskip \noindent $\spadesuit$ \textbf{Unique-supply hat-guessing game \cite{Hat-new}.}
There are $v$ hats, each with a different color in $[v]=\{1,\ldots,v\}$
(so for each color there is only one hat supplied).
The \emph{hat guessing game} is played by $m$ players and one dealer (who is the nature).
\begin{itemize}
\item The dealer randomly places $t$ hats to each player (assume $v-mt=d>0$).\smallskip
\item Each player can observe those hats placed to any other player, but cannot see and has to guess the $t$ colors of hats placed to himself or herself.\smallskip
\item The guess is private between one player and the dealer -- players are forbidden to communicate during the whole game. Yet it is permissible for the players to discuss a strategy before the game starts.\smallskip
\item Player $i~(i\in [m])$ is allowed to guess $g_i$ times. A guess is correct if all the $t$ colors are correct.
   If any guess of any player is correct, all players (as a team) win the game. All the parameters are given before the game starts.\smallskip
\item[Q.] How can we design a strategy to achieve the optimal chance of winning?
\end{itemize}

\begin{example}
$v=3,m=2,t=d=g_1=g_2=1$.
\emph{If Player~1 observes $b$, she guesses $b\mod 3+1$. If Player~2 observes $a$, he guesses $a\mod 3+1$.}
Then, exactly one player guesses right. This strategy wins always and is optimal.
\end{example}

The answer for the two players case (i.e.\ $m=2$) is as follows.
\begin{description}
\item[Graph Model.] Let $A$, $B$ respectively denote the set of colors placed to Player~1 and Player~2. Let $A'=[v]-B$.
The state of the game can be represented as edge $(A,A')$ in $H(v,t)$.
Each player knows one node of the edge; Player~1 knows $A'$ and Player~2 knows $A$.
\smallskip
\item[Upper bound.] The uncertainty for each player is ${t+d\choose d}$. This is the degree of each node.
    By one guess, a player has $1/{t+d\choose d}$ chance to win.
Therefore, the maximum winning probability is no larger than $p=\max\{1,(g_1+g_2)/{t+d\choose d}\}$.\smallskip
\item[Lower bound.] Suppose a 1-factorization of $H(v,t)$ labels each edge by a number in $[{t+d\choose d}]$.
In the $g_1+g_2$ guesses, by respectively choosing the edges with labels $1,\ldots,g_1+g_2$,
  the players win if the label of the edge (state) is in $[g_1+g_2]$, which occurs with probability $p$.
\end{description}

\begin{remark} To play this game, both players wish to have a simple and realistic strategy that is easy to remember.
    This gives us a motivation to design an explicit 1-factorization of $H(v,t)$.
    We also point out that our algorithm for solving P1 and P2 in section~\ref{sect:lexical} find applications
  in this game, because the following task arises in playing the game:
  Given $A$ (or $A'$) and a number $l\in [{t+d\choose d}]$,
    find the unique $A'$ (or $A$) such that $(A,A')$ is labeled with $l$ in the factorization.
\end{remark}
\clearpage
\end{document}